%% file: main_arxiv3.tex
\documentclass{article}
\usepackage{amsmath,amsthm,tikz,caption,subcaption,mathtools}
\usepackage[ruled,vlined,noresetcount]{algorithm2e}
\usetikzlibrary{shapes}
\usepackage{hyperref}       
\usepackage{url}            
\usepackage{booktabs}       
\usepackage{xcolor}
\usepackage{cleveref}
\usepackage{authblk}
\usepackage{float}
\usepackage[normalem]{ulem}
\usepackage{soul}
\usepackage[margin=1in]{geometry}
\usepackage{hyperref}
\hypersetup{colorlinks=true,linkcolor=red,citecolor=magenta,urlcolor=blue}

\input{macros.tex}

\newcommand{\hilb}[1]{\mathcal{H}_{#1}}
\newcommand{\psd}[1]{\mathcal{S}\left( #1 \right)}
\newcommand{\pd}[1]{\mathcal{S}\left( #1 \right)_{++}}

\title{A new operator extension of strong subadditivity of quantum entropy}
\author[1,3]{Ting-Chun Lin}
\author[2]{Isaac H. Kim}
\author[3]{Min-Hsiu Hsieh}

\affil[1]{\textit{Department of Physics, University of California San Diego, CA 92093, USA}}
\affil[2]{\textit{Department of Computer Science, UC Davis,  Davis, CA 95616, USA}}
\affil[3]{\textit{Hon Hai (Foxconn) Research Institute, Taipei, Taiwan.}}
\begin{document}

\sloppy

\maketitle

\begin{abstract}
Let $S(\rho)$ be the von Neumann entropy of a density matrix $\rho$. Weak monotonicity asserts that $S(\rho_{AB}) - S(\rho_A) + S(\rho_{BC}) - S(\rho_C)\geq 0$ for any tripartite density matrix $\rho_{ABC}$, a fact that is equivalent to the strong subadditivity of entropy. We prove an operator inequality, which, upon taking an expectation value with respect to the state $\rho_{ABC}$, reduces to the weak monotonicity inequality. Generalizations of this inequality to the one involving two independent density matrices, as well as their R\'enyi-generalizations, are also presented.
\end{abstract}

\section{Introduction}
\label{sec:introduction}

In quantum mechanics, the notion of conditional probability is generally ill-defined. For example, consider an EPR pair over two qubits. The density matrix of a qubit is maximally mixed but the global state is pure. Thus, the entropy of the global state is strictly smaller than the entropy of its marginal. Examples like this show that one cannot generally ensure
$S(\rho_{AB}) - S(\rho_B)\geq 0$, where $S(\rho) \coloneqq -\text{Tr}(\rho \log \rho)$ is the von Neumann entropy of a density matrix $\rho$.  Nevertheless, the following inequality is still true:
\begin{equation}
S(\rho_{AB})- S(\rho_A) + S(\rho_{BC}) - S(\rho_C) \geq 0. \label{eq:weak_mono}
\end{equation}
This inequality is known as the \emph{weak monotonicity} in the literature. We note that weak monotonicity is equivalent to the strong subadditivity of entropy~\cite{Lieb1973}, a fact that can be shown by considering a purification of $\rho_{ABC}$.

In this paper, we prove operator extensions of Eq.~\eqref{eq:weak_mono}. Consider a tripartite system $\hilb{A}\otimes \hilb{B}\otimes \hilb{C}$. For any positive definite density matrix $\rho_{ABC}$, we show that
\begin{equation}
    \log \rho_{AB}  - \log \rho_A + \log \rho_{BC} - \log \rho_C \leq 0,\label{eq:wm_operator1}
\end{equation}
where a tensor product with the identity operator is suppressed for notational convenience. For instance, $\log \rho_{AB}$ is a short-hand notation for $\log \rho_{AB} \otimes I_C$, where $I_C$ is the identity acting on $\mathcal{H}_C$. Note that, by taking the expectation value with respect to $\rho_{ABC}$, Eq.~\eqref{eq:weak_mono} is recovered. Therefore, Eq.~\eqref{eq:wm_operator1} is an operator extension of weak monotonicity. In fact, this inequality can be extended to an inequality involving two \emph{independent} density matrices $\rho$ and $\sigma$. Let $\rho_{AB}$ and $\sigma_{BC}$ be positive definite density matrices acting on $\hilb{A}\otimes \hilb{B}$ and $\hilb{B} \otimes \hilb{C}$, respectively. We show that
\begin{equation}
    \log \rho_{AB} - \log \rho_A + \log \sigma_{BC}  - \log \sigma_C \leq 0,\label{eq:wm_operator2}
\end{equation}
again suppressing the tensor product with the identity operator.

These inequalities are somewhat surprising because $\log \rho_{AB} - \log \rho_A$ can have positive eigenvalues in general. In particular, in Eq.~\eqref{eq:wm_operator2}, we emphasize again that $\rho$ and $\sigma$ need not be related to each other in any way. While $\log \rho_{AB} - \log \rho_A$ and $\log \sigma_{BC} - \log \sigma_C$ may have positive eigenvalues, their sum, after accounting for the tensor product with the identity, apparently cannot.

The proofs of these inequalities are based on a certain operator inequality involving marginal density matrices and the fact that $f(t) = \ln t$ is an operator monotone function~\cite{Lowner1934,carlen2010trace}. We remark that this operator inequality has some resemblance to an inequality known in the algebraic quantum field theory literature~\cite{Borchers2000,witten2018aps}. However, as we discuss later in Section~\ref{sec:qft}, there are important differences between the two.

Let us make some historical remarks. Since Lieb and Ruskai's seminal proof of strong subadditivity~\cite{Lieb1973}, several strengthenings have appeared in the literature. Carlen and Lieb proved a strengthening which can become nontrivial for entangled quantum states~\cite{Carlen2012}. One of us proved an operator extension~\cite{kim2012operator,ruskai2013remarks}. A strengthening that ensures a robust form of recoverability was proved in Ref.~\cite{Fawzi2015,Wilde2015,Junge2018}. Our operator extension of weak monotonicity can be viewed as yet another strengthening of strong subadditivity. In particular, we reprove, using this new inequality, the operator extension of strong subadditivity~\cite{kim2012operator}; see Corollary~\ref{corollary:operator_extension_ssa}. Thus our new inequality is at least as strong as the operator extension of strong subadditivity.

Another perspective is that we provide an arguably simplest approach to prove strong subadditivity of entropy. Our key observation is that a nontrivial inequality can be obtained by combining Stinespring dilations~\cite{Stinespring1955} of the Accardi-Cecchini coarse graining operator~\cite{accardi1982conditional}. The proof of this inequality is elementary, and once this inequality is obtained, the weak monotonicity follows from an elementary application of L\"owner-Heinz's theorem~\cite{Lowner1934,carlen2010trace} on matrix monotone functions. The strong subadditivity then follows by introducing a purifying system, a fact that is well-known in the literature.
This observation suggests a possibility of deriving new matrix inequalities from dilations of channels.

The rest of the paper is organized as follows. The proofs of our claims (and their generalizations) are presented in Section~\ref{sec:operator_inequality}. In Section~\ref{sec:qft}, we comment on a relation between our inequalities and a similar inequality in quantum field theory. We end with a discussion in Section~\ref{sec:discussion}.

\section{Proofs}
\label{sec:operator_inequality}
Let us begin by first setting up the notation. Let $\mathcal{H}$ be a finite-dimensional Hilbert space. We denote the set of density matrices on $\mathcal{H}$ as $\psd{\mathcal{H}}$. The set of density matrices which are \emph{strictly positive} is denoted as $\pd{\mathcal{H}}$. For simplicity, throughout the paper, we focus on the cases where the density matrices are strictly positive definite. We expect a generalization of our results for positive semi-definite density matrices would require a projection onto an appropriate subspace, which we leave for future work.

Given a density matrix, we shall denote its marginals by specifying the subsystem in the subscript. For instance, $\rho_A$ is a marginal of a density matrix $\rho$ on $\hilb{A}$. We shall denote the operator norm of $M$ as $\| M \|$ and the identity acting on $\hilb{X}$ as $I_{X}$.

Here is the key lemma.
\begin{lemma}
For any $\rho_{AB} \in \pd{\hilb{A}\otimes \hilb{B}}$ and $\sigma_{BC} \in \pd{\hilb{B} \otimes \hilb{C}}$,
\begin{equation}
    \rho_A^{-1}\otimes \sigma_{BC} \leq \rho_{AB}^{-1} \otimes \sigma_C.
\end{equation}
\label{lemma:key}
\end{lemma}
\begin{proof}
Let $\rho_{AB} \in \pd{\hilb{A} \otimes \hilb{B}}$. Consider an operator $V_{A\to ABB^{*}}^{\rho}: \mathcal{H}_A \to \hilb{A} \otimes \hilb{B} \otimes \hilb{B^*} $ defined as follows:
\begin{equation}
    V_{A\to ABB^{*}}^{\rho}\coloneqq \rho_{AB}^{\frac{1}{2}} \rho_A^{-\frac{1}{2}}\sum_{k} |k\rangle_B |k\rangle_{B^*},
\end{equation}
where $\hilb{B^*}$ is an auxiliary Hilbert space such that $\dim\left(\hilb{B}^*\right) = \dim \left(\hilb{B} \right)$, and the summation is taken over a set of orthonormal basis for $\hilb{B^*}$ and $\hilb{B}$. A straightforward calculation shows that ${V_{A\to ABB^{*}}^{\rho}}^{\dagger}V_{A\to ABB^{*}}^{\rho} = \text{Tr}_B (\rho_A^{-\frac{1}{2}}  \rho_{AB} \rho_A^{-\frac{1}{2}})=I_A$. (We remark that, more generally, if $X$ is any operator on $\mathcal{H}_A\otimes \mathcal{H}_B$ and $\hat{X}$ is an operator acting on $\mathcal{H}_A$ such that $\hat{X}|\psi\rangle = \sum_k X|\psi\rangle \otimes |k\rangle_B \otimes |k\rangle_{B^*}$, then $\hat{X}^{\dagger}\hat{X} = \text{Tr}_B(X^{\dagger} X)$.) Thus, $V_{A\to ABB^{*}}^{\rho}$ is an isometry. Similarly, we can define
\begin{equation}
    V_{C\to BB^{*}C}^{\sigma}\coloneqq \sigma_{BC}^{\frac{1}{2}} \sigma_C^{-\frac{1}{2}}\sum_{k} |k\rangle_B |k\rangle_{B^*},
\end{equation}
which is also an isometry.

Let $V_{B\to B'}: \hilb{B} \to \hilb{B'}$ be an isometry, where $\hilb{B'}$ is an auxiliary Hilbert space we use in the following argument. Define $V_{A\to AB'B^*}^{\rho}: \hilb{A} \otimes \hilb{C} \to \hilb{A}\otimes \hilb{B'} \otimes \hilb{B^*} \otimes \hilb{C}$ as follows:
\begin{equation}
    V_{A\to AB'B^*}^{\rho} \coloneqq V_{B\to B'}V_{A\to ABB^*}^{\rho}.
\end{equation}
Consider the operator $(I_A \otimes V_{B\to B'}^{\dagger} \otimes I_C) (I_A \otimes I_{B'} \otimes {V^{\sigma}_{C\to BB^*C}}^{\dagger}) (V^{\rho}_{A\to AB'B^*} \otimes I_B \otimes I_C): \hilb{A} \otimes \hilb{B} \otimes \hilb{C} \to \hilb{A} \otimes \hilb{B} \otimes \hilb{C}$. Since the operator norm of an isometry is $1$, we conclude
\begin{equation}
    \left\|(I_A \otimes V_{B\to B'}^{\dagger} \otimes I_C) (I_A \otimes I_{B'} \otimes {V^{\sigma}_{C\to BB^*C}}^{\dagger}) (V^{\rho}_{A\to AB'B^*} \otimes I_B \otimes I_C) \right\|\leq 1.
\end{equation}
A straightforward calculation shows that
\begin{equation}
\begin{aligned}
    &\phantom{{}={}} (I_A \otimes V_{B\to B'}^{\dagger} \otimes I_C) (I_A \otimes I_{B'} \otimes {V^{\sigma}_{C\to BB^*C}}^{\dagger}) (V^{\rho}_{A\to AB'B^*} \otimes I_B \otimes I_C) \\
    &= \sum_{k, k', k'', k'''}|k'''\rangle_{B}\langle k'''|_{B'} \langle k''|_{B} \langle k''|_{B^*} \sigma_{C}^{-\frac{1}{2}} \sigma_{BC}^{\frac{1}{2}} |k'\rangle_{B'}\langle k'|_{B}\rho_{AB}^{\frac{1}{2}}\rho_A^{-\frac{1}{2}} |k\rangle_B|k\rangle_{B^*} \\
    &= \sum_{k, k', k''}|k'\rangle_{B} \langle k''|_{B} \langle k''|_{B^*} \sigma_{C}^{-\frac{1}{2}} \sigma_{BC}^{\frac{1}{2}} \langle k'|_{B}\rho_{AB}^{\frac{1}{2}}\rho_A^{-\frac{1}{2}} |k\rangle_B|k\rangle_{B^*} \\
    &= \sum_{k, k'}|k'\rangle_{B} \langle k|_{B}  \sigma_{C}^{-\frac{1}{2}} \sigma_{BC}^{\frac{1}{2}} \langle k'|_{B}\rho_{AB}^{\frac{1}{2}}\rho_A^{-\frac{1}{2}} |k\rangle_B
    \\&= \left(\rho_{AB}^{\frac{1}{2}} \otimes \sigma_C^{-\frac{1}{2}}\right) \left(\rho_A^{-\frac{1}{2}} \otimes \sigma_{BC}^{\frac{1}{2}}\right).
\end{aligned}
\label{eq:calculation}
\end{equation}
We therefore obtain
\begin{equation}
\left\| \left(\rho_{AB}^{\frac{1}{2}} \otimes \sigma_C^{-\frac{1}{2}}\right) \left(\rho_A^{-\frac{1}{2}}\otimes \sigma_{BC}^{\frac{1}{2}} \right) \right\|\leq 1
\end{equation}
leading to
\begin{equation}
\left(\rho_A^{-\frac{1}{2}}\otimes \sigma_{BC}^{\frac{1}{2}} \right)   \left( \rho_{AB} \otimes \sigma_C^{-1}\right) \left(\rho_A^{-\frac{1}{2}}\otimes \sigma_{BC}^{\frac{1}{2}} \right)  \leq I_{ABC},
\end{equation}
from which the main claim immediately follows.
\end{proof}

A crucial step of the proof is Eq.~\eqref{eq:calculation}, and here we provide a diagrammatic version of this argument, in order to provide an intuition to the readers. We shall first denote $V_{A\to ABB^*}^{\rho}$ and $V_{C\to BB^*C}^{\sigma}$ as follows:
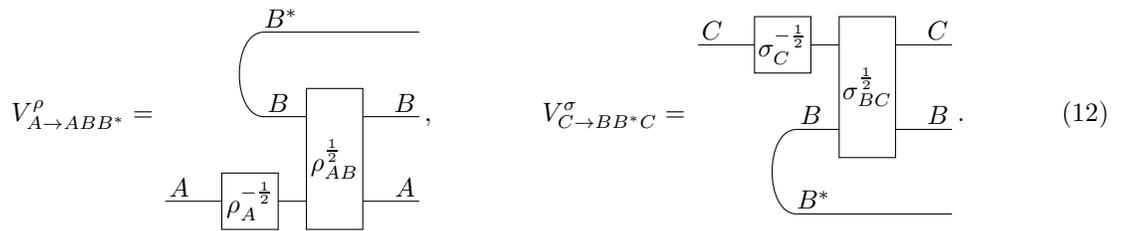
\begin{equation}
    V_{A\to ABB^*}^{\rho} =
    \begin{tikzpicture} [scale=0.75, baseline={([yshift=-.5ex]current bounding box.center)}]
    \draw (0,-1) rectangle node{$\rho_{A}^{-\frac{1}{2}}$} (1,-2) (1.5,0.5) rectangle node{$\rho_{AB}^{\frac{1}{2}}$} (2.5,-2);
    \begin{pgfinterruptboundingbox}
        \begin{scope}[even odd rule]
        \clip (0,-1) rectangle (1,-2) (1.5,0.5) rectangle (2.5,-2) (-2,3) rectangle (4.5,-3);
        \draw (0.75,0) -- (3.5,0) (-1,-1.5) -- (3.5,-1.5);
        \end{scope}
    \end{pgfinterruptboundingbox}
    \draw[out=180,in=180] (0.75,0) to (0.75,1.5);
    \draw (0.75,1.5) -- (3.5,1.5);

    \draw (-0.75,-1.25)node{$A$} (3.25,-1.25)node{$A$} (3.25,0.25)node{$B$} (1,0.25) node{$B$} (1.05,1.75) node{$B^*$};
\end{tikzpicture},
\qquad \qquad
V_{C\to BB^*C}^{\sigma} =
\begin{tikzpicture}[scale=0.75, baseline={([yshift=-.5ex]current bounding box.center)}]
    \draw (0,1) rectangle node{$\sigma_{C}^{-\frac{1}{2}}$} (1,2) (1.5,-0.5) rectangle node{$\sigma_{BC}^{\frac{1}{2}}$} (2.5,2);
    \begin{pgfinterruptboundingbox}
        \begin{scope}[even odd rule]
        \clip (0,1) rectangle (1,2) (1.5,-0.5) rectangle (2.5,2) (-2,-3) rectangle (4.5,3);
        \draw (0.75,0) -- (3.5,0) (-1,1.5) -- (3.5,1.5);
        \end{scope}
    \end{pgfinterruptboundingbox}
    \draw[out=180,in=180] (0.75,0) to (0.75,-1.5);
    \draw (0.75,-1.5) -- (3.5,-1.5);

    \draw (-0.75,1.75)node{$C$} (3.25,1.75)node{$C$} (3.25,0.25)node{$B$} (1.05,-1.25)node{$B^*$} (1,0.25)node{$B$};
\end{tikzpicture}.
\end{equation}
In these diagrams, the input and the output of the maps lie on the left and the right side of the diagram, respectively. Each leg is labeled by the respective subsystem, and the curved leg connecting $B$ and $B^*$ represents $\sum_{k}|k\rangle_B |k\rangle_{B^*}$.

We can similarly represent $V_{A\to AB'B^*}^{\rho}$ as follows:
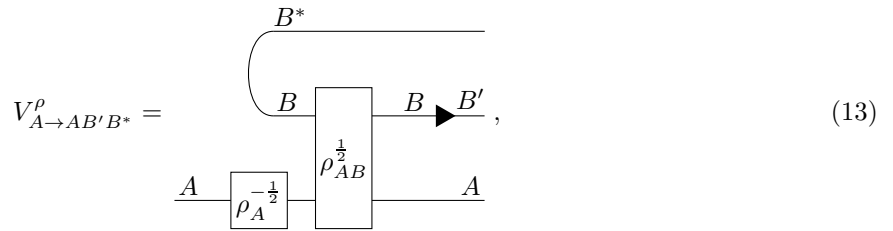
\begin{equation}
V_{A\to AB'B^*}^{\rho} =
    \begin{tikzpicture} [scale=0.75, baseline={([yshift=-.5ex]current bounding box.center)}]
    \draw (0,-1) rectangle node{$\rho_{A}^{-\frac{1}{2}}$} (1,-2) (1.5,0.5) rectangle node{$\rho_{AB}^{\frac{1}{2}}$} (2.5,-2);
    \begin{pgfinterruptboundingbox}
        \begin{scope}[even odd rule]
        \clip (0,-1) rectangle (1,-2) (1.5,0.5) rectangle (2.5,-2) (-2,3) rectangle (4.5,-3);
        \draw (0.75,0) -- (4.5,0) (-1,-1.5) -- (4.5,-1.5);
        \end{scope}
    \end{pgfinterruptboundingbox}
    \draw[out=180,in=180] (0.75,0) to (0.75,1.5);
    \draw (0.75,1.5) -- (4.5,1.5);

    \draw (-0.75,-1.25)node{$A$} (4.25,-1.25)node{$A$} (3.25,0.25)node{$B$} (1,0.25) node{$B$} (1.05,1.75) node{$B^*$};

    \draw (4.25, 0.3) node {$B'$};
    \node[fill=black, regular polygon, regular polygon sides=3, minimum size=0.35cm, inner sep=0pt, rotate=270] () at (3.75, 0) {};

\end{tikzpicture},
\end{equation}
where the triangle corresponds to $V_{B\to B'}$ and the curved leg connecting $B$ and $B^*$ is now $\sum_k \langle k|_B \langle k|_{B^*}$. We can thus obtain
\begin{equation}
\begin{aligned}
    (I_A \otimes I_{B'} \otimes {V^{\sigma}_{C\to BB^*C}}^{\dagger}) (V^{\rho}_{A\to AB'B^*} \otimes I_B \otimes I_C)&=
    \begin{tikzpicture} [scale=0.75, baseline={([yshift=-.5ex]current bounding box.center)}]
    \draw (0,-1) rectangle node{$\rho_{A}^{-\frac{1}{2}}$} (1,-2) (1.5,0.5) rectangle node{$\rho_{AB}^{\frac{1}{2}}$} (2.5,-2);
    \begin{pgfinterruptboundingbox}
        \begin{scope}[even odd rule]
        \clip (0,-1) rectangle (1,-2) (1.5,0.5) rectangle (2.5,-2) (-2,3) rectangle (4.5,-3);
        \draw (0.75,0) -- (4.5,0) (-1,-1.5) -- (4.5,-1.5);
        \end{scope}
    \end{pgfinterruptboundingbox}
    \draw[out=180,in=180] (0.75,0) to (0.75,1.5);
    \draw (0.75,1.5) -- (4.5,1.5);

    \draw (-0.75,-1.25)node{$A$} (4.25,-1.25)node{$A$} (3.25,0.25)node{$B$} (1,0.25) node{$B$} (1.1,1.75) node{$B^*$};

    \draw (4.25, 0.3) node {$B'$};
    \node[fill=black, regular polygon, regular polygon sides=3, minimum size=0.35cm, inner sep=0pt, rotate=270] () at (3.75, 0) {};
    \begin{scope}[xshift=8cm, xscale=-1, yshift=3cm]
    \draw (0,1) rectangle node{$\sigma_{C}^{-\frac{1}{2}}$} (1,2) (1.5,-0.5) rectangle node{$\sigma_{BC}^{\frac{1}{2}}$} (2.5,2);
    \begin{pgfinterruptboundingbox}
        \begin{scope}[even odd rule]
        \clip (0,1) rectangle (1,2) (1.5,-0.5) rectangle (2.5,2) (-2,-3) rectangle (4.5,3);
        \draw (0.75,0) -- (3.5,0) (-1,1.5) -- (3.5,1.5);
        \end{scope}
    \end{pgfinterruptboundingbox}
    \draw[out=180,in=180] (0.75,0) to (0.75,-1.5);
    \draw (0.75,-1.5) -- (3.5,-1.5);

    \draw (-0.75,1.75)node{$C$} (3.25,1.75)node{$C$} (3.25,0.25)node{$B$} (0.95,-1.25)node{$B^*$} (1,0.25)node{$B$};
    \end{scope}
\end{tikzpicture}
\\[15pt]
&=
\begin{tikzpicture} [scale=0.75, baseline={([yshift=-.5ex]current bounding box.center)}]
\node[fill=black, regular polygon, regular polygon sides=3, minimum size=0.35cm, inner sep=0pt, rotate=270] () at (3.5, 0) {};
    \draw (-0.75, -1.25) node {$A$} (-0.75, 0.25) node {$B$} (-0.75, 1.75) node {$C$};
    \draw (3, -1.25) node {$A$} (3, 0.25) node {$B$} (3, 1.75) node {$C$};
    \draw (4, 0.27) node {$B'$};
    \draw[] (-1, 1.5) -- ++ (5.5,0);
    \draw[] (-1, 0) -- ++ (5.5,0);
    \draw[] (-1, -1.5) -- ++ (5.5,0);
    \draw[fill=white] (0,2) rectangle (1, -0.5);
    \draw[fill=white] (1.5, 2) rectangle (2.5, 1);
    \draw[fill=white] (0, -1) rectangle (1, -2);
    \draw[fill=white] (1.5, 0.5) rectangle (2.5, -2);
    \draw (0,2) rectangle node{$\sigma_{BC}^{\frac{1}{2}}$} (1,-0.5)
                  (1.5,2) rectangle node{$\sigma_{C}^{-\frac{1}{2}}$} (2.5,1)
                  (0,-1) rectangle node{$\rho_{A}^{-\frac{1}{2}}$} (1,-2)
                  (1.5,0.5) rectangle node{$\rho_{AB}^{\frac{1}{2}}$} (2.5,-2);

\end{tikzpicture},
\end{aligned}
\end{equation}
where the second line is obtained by simply ``straigtening out'' the curved leg. At this point, it is straightforward to see that $\left(\rho_{AB}^{\frac{1}{2}} \otimes \sigma_C^{-\frac{1}{2}}\right) \left(\rho_A^{-\frac{1}{2}} \otimes \sigma_{BC}^{\frac{1}{2}}\right)$ can be obtained by applying the inverse of $V_{B\to B'}^{\dagger}$, which completes the argument.

\begin{remark}
The isometry $V^{\rho}_{A\to ABB^*}$ is the Stinespring dilation~\cite{Stinespring1955} of the Accardi-Cecchini coarse graining operator~\cite{accardi1982conditional}.
\end{remark}

By the L\"owner-Heinz theorem~\cite{Lowner1934,carlen2010trace}, $f(t) = \log t$ is operator monotone. Thus, we immediately obtain the following result.
\begin{theorem}
\label{thm:main}
    For any $\rho_{AB} \in \pd{\hilb{A} \otimes \hilb{B}}$ and $\sigma_{BC} \in \pd{\hilb{B} \otimes \hilb{C}}$,
    \begin{equation}
        \log \rho_{AB} - \log \rho_A + \log \sigma_{BC}  - \log \sigma_C \leq 0.
    \end{equation}
\end{theorem}
\noindent
We remark that, by taking $\rho_{AB}$ and $\sigma_{BC}$ as the marginal density matrices of $\rho_{ABC}$, Eq.~\eqref{eq:wm_operator1} follows. Moreover, by taking an expectation value with respect to $\rho_{ABC}$, weak monotonicity, and subsequently, the strong subadditivity of entropy~\cite{Lieb1973} follows as well.

Moreover, Theorem~\ref{thm:main} implies the operator extension of strong subadditivity~\cite{kim2012operator,ruskai2013remarks}.
\begin{corollary}
\label{corollary:operator_extension_ssa}
    For any $\rho_{ABC} \in \pd{\hilb{A} \otimes \hilb{B} \otimes \hilb{C}}$,
    \begin{equation}
        \textnormal{Tr}_{BC}\left(\rho_{ABC} (\log \rho_{ABC} + \log \rho_B - \log \rho_{AB} - \log \rho_{BC}) \right) \geq 0.
    \end{equation}
\end{corollary}
\begin{proof}
Consider a purification of $\rho_{ABC}$, denoted as $|\rho\rangle_{ABCD}$, where $D$ is the purifying space. By Theorem~\ref{thm:main},
\begin{equation}
    \log \rho_{BC} - \log \rho_B + \log \rho_{CD}  - \log \rho_D \leq 0.
\end{equation}
For any $M_A$ acting on $\hilb{A}$,
\begin{equation}
    \langle \rho| (M_A^{\dagger}\otimes I_{BCD})(\log \rho_{BC} + \log \rho_{CD} - \log \rho_B - \log \rho_D) (M_A\otimes I_{BCD}) |\rho\rangle \leq 0.
\end{equation}
Since $\log \rho_{CD} |\rho\rangle = \log \rho_{AB} |\rho\rangle$ and $\log \rho_D |\rho\rangle = \log \rho_{ABC} |\rho\rangle$, we get
\begin{equation}
\text{Tr}_A\left(M_A^{\dagger}M_A\text{Tr}_{BC}\left(\rho_{ABC} (\log \rho_{ABC} + \log \rho_B - \log \rho_{AB} - \log \rho_{BC}) \right) \right) \geq 0,
\end{equation}
which implies the claim.
\end{proof}

The L\"owner-Heinz theorem also implies that $f(t) = t^{\alpha}$ is operator monotone for $\alpha\in [0, 1]$. Thus, the following theorem also follows, which can be viewed as a R\'enyi generalization of Theorem~\ref{thm:main}.
\begin{theorem}
    For any $\rho_{AB} \in \pd{\hilb{A} \otimes \hilb{B}}$ and $\sigma_{BC} \in \pd{\hilb{B} \otimes \hilb{C}}$,
    \begin{equation}
        \rho_{A}^{-\alpha} \otimes \sigma_{BC}^{\alpha} \leq \rho_{AB}^{-\alpha} \otimes \sigma_C^{\alpha},
    \end{equation}
    for $\alpha \in [0, 1]$.
\end{theorem}

\section{A related inequality from algebraic quantum field theory}
\label{sec:qft}

We remark that there is a known result in the algebraic quantum field theory literature~\cite{Borchers2000} which appears similar to Lemma~\ref{lemma:key}. This result dates back to the work of Petz~\cite{petz1986quasi}, which was used to prove the data processing inequality. Here we introduce this result and comment on this similarity. (An introduction to von Neumann algebra and related concepts can be found in Ref.~\cite{witten2018aps}.) Let $\hilb{}$ be a Hilbert space. Let $|\Psi\rangle\in \mathcal{H}$ be a cyclic and separating vector for a von Neumann algebra $\mathcal{A}$ on $\mathcal{H}$. Let $|\Psi\rangle, |\Phi\rangle\in \mathcal{H}$ be vectors for a von Neumann algebra $\mathcal{A}$ on $\mathcal{H}$ where $|\Psi\rangle$ is cyclic and separating over $\mathcal{A}$. Then we can define a \emph{relative modular operator}~\cite{TomitaTakesaki} as $\Delta_{\Psi |\Phi;\mathcal{A}}=S_{\Psi|\Phi;\mathcal{A}}S_{\Psi|\Phi;\mathcal{A}}^{\dagger}$, where $S_{\Psi|\Phi}$ is an anti-linear operator such that for any $\mathsf{a}\in \mathcal{A}$,
\begin{equation}
    S_{\Psi|\Phi;\mathcal{A}} \mathsf{a} |\Psi\rangle =  \mathsf{a}^{\dagger}|\Phi\rangle.
\end{equation}
Let $\mathcal{A}_1$ be an algebra. It is known that, for any algebra $\mathcal{A}_2 \subset \mathcal{A}_1$, the following inequality holds:
\begin{equation}
\Delta_{\Psi|\Phi;\mathcal{A}_2} \geq\Delta_{\Psi|\Phi;\mathcal{A}_1},\label{eq:ineq_qft}
\end{equation}
which can be found in
~\cite[Equation (2.1.3)]{Borchers2000}, and more recently,~\cite[Equation (3.36)]{witten2018aps}. This inequality makes sense only if both sides are well-defined, which requires $\Psi$ to be cyclic and separating for both $\mathcal{A}_1$ and $\mathcal{A}_2$.

To show the similarity and the difference between Eq.~\eqref{eq:ineq_qft} and Lemma~\ref{lemma:key}, let us consider the following plausible but incorrect argument to prove Theorem~\ref{thm:main}. Let $|\Psi\rangle, |\Phi\rangle \in \hilb{A} \otimes \hilb{B} \otimes \hilb{C} \otimes \hilb{D}$, where $\hilb{A}, \hilb{B}, \hilb{C}$, and $\hilb{D}$ are finite-dimensional Hilbert spaces. Let $|\Psi\rangle$ be a purification of $\rho_{AB}$ and $|\Phi\rangle$ be a purification of $\sigma_{BC}$, both assumed to be of full rank. Define the following algebras:
\begin{equation}
    \begin{aligned}
    \mathcal{A}_1 &= \{I_{A}\otimes M_{BCD}: M_{BCD}\in \mathcal{B}(\hilb{B} \otimes \hilb{C} \otimes \hilb{D}) \}, \\
    \mathcal{A}_2 &= \{I_{AB}\otimes M_{CD}: M_{CD} \in \mathcal{B}( \hilb{C} \otimes \hilb{D})\},
    \end{aligned}
\end{equation}
where $\mathcal{B}(\hilb{})$ is the space of bounded operators acting on $\hilb{}$.

If we can find $|\Psi\>$ which is cyclic and separating for both $\cA_1$ and $\cA_2$, following~\cite[Sec. 4]{witten2018aps}, the relative modular operators become
    \begin{equation}
        \Delta_{\Psi|\Phi;\mathcal{A}_2} = \rho_{AB}^{-1} \otimes \sigma_{CD}, \qquad
        \Delta_{\Psi|\Phi;\mathcal{A}_1} = \rho_{A}^{-1} \otimes \sigma_{BCD}. \label{eq:wrong-tmp}
    \end{equation}
If Eq.~\eqref{eq:wrong-tmp} is correct, we could use Eq.~\eqref{eq:ineq_qft} and take a partial trace on $D$ over both sides, obtaining
    \begin{equation}
    \rho_{AB}^{-1} \otimes \sigma_C \geq \rho_{A}^{-1}\otimes  \sigma_{BC},
    \end{equation}
which is exactly Lemma~\ref{lemma:key}.

Unfortunately, such $|\Psi\>$ does not exist in general when the Hilbert spaces are finite-dimensional.  This is because $\dim \cH_A = \dim (\cH_B \otimes \cH_C \otimes \cH_D)$ when $|\Psi\>$ is cyclic and separating for $\cA_1$ and $\dim (\cH_A \otimes \cH_B) = \dim (\cH_C \otimes \cH_D)$ when $|\Psi\>$ is cyclic and separating for $\cA_2$. The two conditions cannot simultaneously hold in general for finite-dimensional Hilbert spaces unless $\dim \cH_B = 1$. Interestingly, this issue does not arise in certain states of quantum field theory. For instance, any vacuum state is both cyclic and separating for any field algebra associated to an open set of the Minkowski space, thanks to the Reeh-Schlieder theorem \cite{witten2018aps,reeh1961bemerkungen}.
One may hope to circumvent this issue of cyclic and separating condition by considering a more general definition of relative modular operators that does not require the state to be cyclic or separating \cite[Appendix A]{ceyhan2020recovering}. However, it is then not obvious if Eq.~\eqref{eq:ineq_qft} is true because under such definition, it is not clear if $S_{\Psi|\Phi;\cA_1}$ is an extension of $S_{\Psi|\Phi;\cA_2}$.

\section{Discussion}
\label{sec:discussion}
In this paper, we proved an operator extension of weak monotonicity. It is interesting to note that our argument also leads to yet another proof of strong subadditivity~\cite{Lieb1973}. What is notable about this new proof is that the strong subadditivity is proved by first proving the weak monotonicity, not the other way around. The key observation was Lemma~\ref{lemma:key}, which followed immediately from constructions of certain isometries. We leave it as an open problem to explore the consequences of this simple but powerful observation.

\section*{Acknowledgement}
IK thanks Mark Wilde for helpful discussions. MHH thanks Marco Tomamichel for helpful discussions. TCL thanks John McGreevy, Bowen Shi, and Xiang Li for helpful discussions. We thank Geoff Penington for pointing out Ref.~\cite{ceyhan2020recovering} and Andreas Winter for helpful comments. We thank the anonymous reviewers for their helpful comments and corrections to the citations.

\bibliographystyle{unsrt}
\bibliography{references.bib}

\end{document}

%% file: macros.tex
\newtheorem{theorem}{Theorem}
\newtheorem{corollary}{Corollary}

\newtheorem{lemma}{Lemma}

\newtheorem{remark}[theorem]{Remark}

\newcommand{\cA}{\mathcal{A}}

\newcommand{\cH}{\mathcal{H}}

\newcommand{\interior}[1]{%
  {\kern0pt#1}^{\mathrm{o}}%
}

\renewcommand{\>}{\rangle}
